\documentclass[pra,10pt,twocolumn,superscriptaddress,floatfix,showpacs]{revtex4-1}
\usepackage[utf8]{inputenc}  
\usepackage[T1]{fontenc}     
\usepackage{lmodern}
\usepackage[british]{babel}  
\usepackage[scaled=0.86]{berasans}  
\usepackage[colorlinks=true,allcolors=blue]{hyperref}  
\usepackage{graphicx} 
\usepackage[babel]{microtype}  
\usepackage{amsmath,amssymb,amsthm,bm,amsfonts,mathrsfs,bbm} 
\usepackage{setspace}
\usepackage{csquotes}

\usepackage{xspace}  
\usepackage{pgfplots}
\usepackage{verbatim}

\newcommand\id{\leavevmode\hbox{\small1\kern-3.3pt\normalsize1}}

\newcommand{\s}{{\bf s}}

\newcommand{\tr}{\mbox{Tr}}

\newtheorem{theorem}{Theorem}

\newtheorem{corollary}[theorem]{Corollary}

\newtheorem{lemma}[theorem]{Lemma}

\begin{document}
\title{\textbf{Absolute Non-Violation of a Three-Setting Steering Inequality by Two-Qubit States} }

\author{Some Sankar Bhattacharya}
\email{somesankar@gmail.com}
\affiliation{Physics and Applied Mathematics Unit, Indian Statistical Institute, 203 B. T. Road, Kolkata 700108, India.}

\author{Amit Mukherjee}
\affiliation{Physics and Applied Mathematics Unit, Indian Statistical Institute, 203 B. T. Road, Kolkata 700108, India.}

\author{Arup Roy}
\affiliation{Physics and Applied Mathematics Unit, Indian Statistical Institute, 203 B. T. Road, Kolkata 700108, India.}

\author{Biswajit Paul}
\affiliation{Department of Mathematics, South Malda College, Malda, West Bengal, India}

\author{Kaushiki Mukherjee}
\affiliation{Department of Mathematics, Government Girls’ General Degree College, Ekbalpore, Kolkata, India.}

\author{Indranil Chakrabarty}
\affiliation{Center for Security, Theory and Algorithmic Research, International Institute of Information 
Technology-Hyderabad, Gachibowli, Telangana-500032, India.}

\author{C. Jebaratnam}
\affiliation{S. N. Bose National Centre for Basic Sciences, Salt Lake, Kolkata 700 098, India.}

\author{Nirman Ganguly}
\affiliation{Physics and Applied Mathematics Unit, Indian Statistical Institute, 203 B. T. Road, Kolkata 700108, India.}
\email{nirmanganguly@gmail.com}
\thanks{ At present on leave from Department of Mathematics, Heritage Institute of Technology, Kolkata-107,India.}

\begin{abstract}
 Steerability is a characteristic nonlocal trait of quantum states lying in between entanglement and Bell nonlocality. A given quantum state is considered to be steerable if it violates a suitably chosen steering inequality. A quantum state which otherwise satisfies a certain inequality can violate the inequality under a global change of basis i.e, if the state is transformed by a nonlocal unitary operation. Intriguingly there are states which preserve their non-violation(pertaining to the said inequality) under any global unitary operation. The present work explores the effect of global unitary operations on the steering ability of a quantum state which live in two qubits. We characterize such states in terms of a necessary and sufficient condition on their spectrum. Such states are also characterized in terms of some analytic characteristics of the set to which they belong.  Looking back at steerability the present work also provides a relation between steerability and quantum teleportation together with the derivation of a result related to the optimal violation of steering inequality . An analytic estimation of the size of such non-violating states in terms of purity is also obtained. Interestingly the estimation in terms of purity also gives a necessary and sufficient condition in terms of bloch parameters of the state.  Illustrations from some signature class of quantum states further underscore our observations. 
\end{abstract}
\date{\today}
\maketitle

\section{Introduction}

The ubiquitous roles of entanglement and non-locality\cite{Review Ent,Review NL,Bell64,CHSH69}are exemplified in many quantum computation and information processing tasks\cite{Bennett93,Bennett92,random, key,dw,game}. These two inequivalent features have been the cornerstone of Quantum Mechanics.\\ 
\indent In 1935 Einstein, Podolsky and Rosen contributed an argument claiming the incompleteness of Quantum Mechanics \cite{epr}.In the following year Schr\"{o}dinger envisaged the celebrated concept of `steering'\cite{schr}. Recently, Wiseman \emph{et al.} have developed this phenomena in the form of a task\cite{wisemanprl07,wisemanpra07}.
They argued that steering refers to the scenario where one party usually called Alice, wishes to convince the other party (commonly referred to as Bob) that she can steer or construct the conditional state on Bob’s side by making measurements on her part. This kind of interpretation of the phenomena have induced great interest in foundational research in recent times\cite{jevtic2014,rudolph2014,jevtic12014,brunneroneway2014,wittman2014,wiseman2014}.\\
\indent Apart from the foundational interest, the study of steering also finds applications in one-sided device independent scenario where only one party trusts his/her quantum device but the other party’s device is untrusted. As a concrete example it has been shown that steering
allows for secure quantum key distribution when one of the parties’ device cannot be trusted. One big advantage in this direction is that such scenarios are experimentally less
demanding than fully device-independent protocols (where both of the parties distrust their devices) and, at the same time, require less assumptions than standard quantum cryptographic scenarios .\\
\indent In 1964, Bell sought a way to demonstrate that certain correlations appearing in quantum
mechanics are incompatible with the notions of locality and reality a.k.a local-realism, through
an inequality involving measurement statistics. In 1969, Clauser–Horne–Shimony–Holt
(CHSH) proposed a set of simple Bell inequalities which are easy to realize experimentally
. In the same spirit of Bell’s inequality in nonlocality, several steering inequalities (SIs)
have been proposed, so that a violation of any such SI can render a state to be
steerable. 
To test EPR steering Reid first proposed a testable formulation for continuous-variable systems based on position-momentum uncertainty relation \cite{reid} which was experimentally tested by Ou et al \cite{ou}. Cavalcanti \emph{et al.} developed a general construction of experimental EPR-steering criteria based on the assumption of existence of LHS(Local Hidden state) model \cite{caval}. Importantly this general construction is applicable to both  discrete as well as continuous-variable observables and Reid's criterion appears as a special case of this general formulation.\\
\indent One of the most intriguing problem in quantum mechanics is to find the signature of entanglement\cite{gurvits1}.The necessary and sufficient criteria for identifying entanglement in lower dimensions($2 \otimes 2$ and $2 \otimes 3$) \cite{peres,horodecki} through negativity of the partial transpose fails in higher dimensions due to the presence of PPT bound entangled state\cite{bound}. However an effective method can be provided for detection of entanglement via entanglement witness\cite{horodecki,terhal,review}. Entanglement witnesses $W_E$ are hermitian operators having at least one negative eigenvalue which satisfy the inequalities (i) $Tr(W_E \varrho_{sep}) \geq 0, \forall$ separable states $\varrho_{sep}$ and (ii) $Tr(W_E\varrho_{ent}) < 0$ for at at least one entangled state $\varrho_{ent}$. The geometric form of the Hahn-Banach theorem states that points lying outside of a convex and closed state can be separated from it by a hyperplane\cite{holmes} which effectively tells the existence of entanglement witnesses for detection of entanglement\cite{horodecki}. As entanglement witnesses are hermitian they have been conveniently used in experimental detection of entanglement \cite{barbieri,wiec}. Method of constructing witness to detect entanglement has been further extended to detect useful resources for teleportation\cite{telwit,telwit2,telwitcom}.\\ 
\indent An analogous procedure to capture non-locality is through a Bell-CHSH witness \cite{bellwit}.For two qubits, a quantum states $\rho$ does not violate the Bell-CHSH inequality iff $M(\rho) \leq 1$, where $M(\rho)$ is defined as the sum of the two largest eigenvalues of the matrix $T_{\rho}^{t}T_{\rho}$, $T_{\rho}$ being the correlation matrix in the Hilbert-Schmidt representation of $\rho$. Thus, $M(\rho) > 1$ is a signature of the non-locality of the state\cite{horobell}.\\ 
\indent Pertaining to separability of quantum states, questions have been raised on the characterization of absolutely separable\cite{vers,johnston} and absolutely PPT states\cite{absppt}. Precisely, a quantum state which is entangled(respectively PPT) in some basis might not be entangled(resp. PPT) in some other basis. This depends on the factorizability of the underlying Hilbert space. Thus, the characterization of states which remain separable(resp. PPT) under any factorization of the basis is pertinent \cite{vers,johnston,absppt,witabs}. Some of the authors have addressed the problem of a state being absolutely Bell-CHSH local \cite{bellGU1,bellGU2}. A state is termed as absolutely Bell-CHSH local if it remains local with respect to the CHSH inequality , under any global unitary operation \cite{bellGU1,bellGU2}. The effect of global unitary operations on the conditional entropy of a two-qubit system has also been probed recently \cite{patro}.\\ 
\indent In the present work, we address the question of non-violation of a steering inequality under any global unitary operation. This is understood that a pure product state can be converted to a maximally entangled state by a suitable global unitary operation which thereby can violate the steering inequality. However, if purity of the initial state goes below a certain extent, the state cannot be made to violate the steering inequality even by a global unitary operation. The characterization of such states forms one of the main constituents of the present submission.With some specific steering inequalities we have characterized such states which retain their non-violating character under the action of global unitary. We further find a criterion to determine whether a particular state exhibits absolute non-violation. This is done by the derivation of a result that the maximal steering violations w.r.t the inequality we have considered is attained at the respective Bell diagonal state for a fixed spectrum. The size of such states are estimated with illustrations to support our observations.\\
\indent Starting with some prerequisites for our current work, we also observe some distinctive characteristics of a steerable state. While reviewing the basic concepts needed for our work we have derived some new results pertaining to the optimal violation of a steering inequality and the relation of steerability to teleportation.\\
\indent The paper is arranged as follows: In section(\ref{sec:pre}) we discuss some prerequisites needed for our study and also observe some distinctive traits of steerable states. In section(\ref{sec:charac}) we derive the conditions for a state to be absolutely non-violating(w.r.t our chosen inequality) and find an estimation of the size of such states. Illustrations are provided in section (\ref{sec:illus}).Finally we conclude in section (\ref{sec:concl}).
\section{Prerequisites, Few Definitions and Some new Results}\label{sec:pre}

\subsubsection{Notations and Definitions} \label{def}

At the outset we put together some notations and definitions to be followed in our analysis.$ \mathfrak{B}(X) $ denotes the set of bounded linear operators acting on X. The density matrices that 
we consider here, are operators acting on two qubits, i.e., 
$\rho \in \mathfrak{B}(H_{2} \otimes H_{2})$. $\mathbf{Q}$ denotes the set of all density matrices.\\
In \cite{caval} authors have developed a series of steering inequalities to check whether a bipartite state is steerable when both the parties are allowed to perform $\mathit{n}$  measurements on his or her part, which is given by the following equation
\begin{equation}\label{steering n}
F_n(\rho,\mu) =\frac{1}{\sqrt{n}} \Big|\sum_{i=1}^n\langle A_i\otimes B_i\rangle\Big| \leqslant 1,
\end{equation}

The inequalities for 
$\mathit{n}=2,3$ are as given below:
\begin{equation}\label{steering 2}
F_2(\rho,\mu) =\frac{1}{\sqrt{2}} \Big|\sum_{i=1}^2\langle A_i\otimes B_i\rangle\Big| \leqslant 1,
\end{equation}

\begin{equation}\label{steering 3}
F_3(\rho,\mu) =\frac{1}{\sqrt{3}} \Big|\sum_{i=1}^3\langle A_i\otimes B_i\rangle\Big| \leqslant 1,
\end{equation}
where $A_i = \hat{u}_i\cdot\vec{\s}$,  $B_i = \hat{v}_i\cdot\vec{\s}$, $\vec{\s}=(\s_1,\s_2,\s_3)$ is a vector composed of the Pauli matrices, $\hat{u}_i\in\mathbb{R}^3$ are unit vectors, $\hat{v}_i\in\mathbb{R}^3$ are orthonormal vectors, $\mu=\{\hat{u}_1,\cdots,\hat{u}_n,\hat{v}_1,\cdots,\hat{v}_n\}$ is the set of measurement directions, $\langle A_i\otimes B_i\rangle=\text{Tr}(\rho A_i\otimes B_i)$, and $\rho\in\mathcal{H}_A\otimes\mathcal{H}_B$ is some bipartite quantum state. Pertaining to inequality (\ref{steering n}), one may construct steerability witnesses $ S^W $ \cite{caval}. Precisely $ Tr(S^W \chi) <0 $ for at least one steerable state $ \chi $ and $ Tr(S^W\rho) \ge 0 $ for all unsteerable states $ \rho $.\\
  We denote by $ \mathbf{L} $ and $\mathbf{US}$, the set of all states which do not violate Bell-CHSH and steering inequality (\ref{steering n}) respectively.
We denote respectively by $\mathbf{AL}$,$\mathbf{AUS} $ as the sets containing states which do not violate the Bell-CHSH inequality and the steering inequality (\ref{steering n}) under any global unitary operation ($ U $). One can easily verify that $\mathbf{AUS}$ forms a non-empty subset of $\mathbf{US}$, as $\frac{1}{4}(I \otimes I) \in \mathbf{AUS}$. The set of states that do not violate inequality (\ref{steering 3}) under any global unitary operations will be denoted by $ \mathbf{AUS}_3 $.

\subsubsection{ $\mathbf{AUS}$ is convex and compact}\label{compact}

A formal characterization follows below,
\begin{theorem}  $\mathbf{US}$ is a convex and compact subset of  $\mathbf{Q}$.
\end{theorem}
\begin{proof}
	First note that the statements below are equivalent:\\
	(i) $\rho \in \mathbf{US}$\\
	(ii) $ \forall $  steering operator $ S , Tr(S\rho) \leq 1$\\
	(iii) $ \forall $  steering witness $ S^{W}, Tr(S^{W}\rho) \geq 0 $\\
	In view of the above, we can rewrite $ \mathbf{US} $ as, $ \mathbf{US} = \lbrace \rho: Tr(S^{W}\rho) \geq 0, \forall S^{W}   \rbrace $. Now consider a function $ f_{1}: \mathbf{Q} \rightarrow \mathbb{R} $, defined as
	\begin{equation}
	f_{1}(\chi)=Tr(S^{W_{1}}\chi)
	\end{equation}
	where, $S^{W_{1}}$ is a fixed steering witness. Let $ US_{1}= \lbrace \chi_{1}:Tr(S^{W_{1}}\chi_{1}) \geq 0 \rbrace $.  $ Tr(S^{W_{1}}\chi_{1})  $ will have a maximum value $ d_{1}$ (say). Therefore, one may write $ US_{1} = f_{1}^{-1}[0,d_{1}] $. $ f_{1} $ is a continuous function as $ Tr $ is a continuous function. This in turn implies $ US_{1} $ is a closed set. Continuing as above, one may define $ US_{i} $ for a fixed $S^{W_{i}}  $. $ US_{i} $ will be closed $ \forall  i $. Since, arbitrary intersection of closed sets is closed, $ \bigcap_{i}US{i} $ is closed. It is easy to see that $ \bigcap_{i}US{i} = \mathbf{US} $ . Hence, $ \mathbf{US} $ is closed. If we now take two arbitrary $ \rho_{1},\rho_{2} \in \mathbf{US}$ , then $ Tr[S^{W}(\lambda \rho_{1}+(1-\lambda)\rho_{2})] \geq 0$ for any $ S^{W}  $, $ \lambda \in [0,1] $. This follows from the fact that $ Tr[S^{W}\rho_{i}] \geq 0 , i=\lbrace 1,2 \rbrace$ for any $ S^{W} $.Thus $ \mathbf{US} $ is convex. Since $ \mathbf{Q} $ is compact, $ \mathbf{US} $ being a closed subset of $ \mathbf{Q} $, is thus  compact. Hence the theorem.
\end{proof}
This theorem facilitates the characterization of the set $ \mathbf{AUS} $ as stated in the theorem below:\\
\begin{theorem} $\mathbf{AUS}$ is a convex and compact subset of $\mathbf{US}$.
\end{theorem}
\begin{proof}
	We only show that $ \mathbf{AUS} $ is convex as the compactness follows from a retrace of the steps presented in \cite{witabs}. \\
	Take two arbitrary $ \sigma_{1}, \sigma_{2} \in \mathbf{AUS} $. One may rewrite $ \mathbf{AUS}= \lbrace \sigma: Tr[S^{W}(U \sigma U^{\dagger})] \geq 0, \forall S^{W} , \forall U \rbrace $. Therefore, for any $ U $, $ U[\lambda \sigma_{1} + (1-\lambda) \sigma_{2}]U^{\dagger} = \lambda \sigma_{1}^{\prime} + (1-\lambda)\sigma_{2}^{\prime} \in \mathbf{AUS} $. This follows, since $ \mathbf{US} $ is convex. [$ \sigma_{i}^{\prime}=U \sigma_{i} U^{\dagger} $].\\
	Hence, the theorem.
\end{proof}
The above  characterization enables one to formally define an operator($ W^{S} $) which detects states that violate steering inequality under global unitary. 
\begin{eqnarray}
Tr(W^{S}\sigma) \geq 0 , \forall \sigma \in \mathbf{AUS}\label{ineq1}\\
\exists \varsigma \in \mathbf{US}-\mathbf{AUS} , Tr(W^{s}\varsigma) < 0 \label{ineq2}
\end{eqnarray}
Consider $\varsigma \in \mathbf{US}-\mathbf{AUS}$. There exists a unitary operator $U_{e}$ such that $U_{e}\varsigma U_{e}^{\dagger}$ violates steering inequality. Consider a 
steering witness $S^W$ that detects $U_{e}\varsigma U_{e}^{\dagger}$, i.e., $Tr(S^W U_{e}\varsigma U_{e}^{\dagger}) < 0$. Using the cyclic property of the trace, one  obtains $Tr(U_{e}^{\dagger}S^W U_{e}\varsigma ) < 0$. We thus claim that  
\begin{equation}
W^S=U_{e}^{\dagger}S^W U_{e}
\label{witop}
\end{equation}
is our desired operator. To see that it satisfies inequality (\ref{steering n}), we consider its action on a state $\sigma$ from $\mathbf{AUS}$. 
We have $Tr(W^S \sigma)=Tr(U_{e}^{\dagger}S^W U_{e} \sigma)=Tr(S^W U_{e} \sigma U_{e}^{\dagger} )$. As $\sigma $ $ \in  \mathbf{AUS} $,and $ S^W $ is a steering witness $ Tr(S^W U_{e} \sigma U_{e}^{\dagger} ) \geq 0$. This implies that $W$ has a non-negative expectation value on all  states $\sigma \in \mathbf{AUS}$. \\

\subsubsection{Steerability and Teleportation}

In Hilbert-Schmidt representation any density matrix living in two qubits can be expressed as ,
\begin{equation}
\zeta = \frac{1}{4}(I \otimes I + \overrightarrow{a}.\overrightarrow{s} . \otimes I + I \otimes \overrightarrow{b}.\overrightarrow{s} + \underset{i,j}{\sum} t_{ij} s_i \otimes s_j)
\label{bloch}
\end{equation}
$ \overrightarrow{a} , \overrightarrow{b} $, being the local bloch vectors and $ T=[t_{ij}] $ is the correlation matrix.\\
 In \cite{horodecki1996}, a necessary and sufficient condition for a state to be useful for teleportation was derived in terms of bloch parameters. Namely, a state $ \zeta $ is useful for teleportation iff \cite{horodecki1996},
 \begin{equation}
 N(\zeta)=Tr \sqrt{T^\dagger T}= \sum\limits_{i=1}^{3} \sqrt{u_i} > 1
 \end{equation}
 where $ u_i $ are the eigenvalues of $ T^\dagger T $ .\\
 In \cite{angelo2016} the authors derived a steerability measure under 3 measurement settings for density matrices of the form ,
\begin{equation}
\zeta^\prime = \frac{1}{4}(I \otimes I + \overrightarrow{a^\prime}.\overrightarrow{s} . \otimes I + I \otimes \overrightarrow{b^\prime}.\overrightarrow{s} + \underset{i}{\sum} c_i s_i \otimes s_i)\label{diagonalT}
\end{equation}
where the correlation matrix of $ \zeta^\prime $ is $ T^\prime = diag(c_1,c_2,c_3) $. The steerability measure was given by by $ F_3(\zeta^\prime) = \sqrt{Tr(T^{\prime^\dagger}T^\prime)} $(subscript 3 denotes 3 measurement settings). Since the density matrices given in equations (\ref{bloch}) and (\ref{diagonalT}) are local unitary equivalent, we have,
\begin{equation}
F_3(\zeta^\prime) = \sqrt{Tr(T^{\prime^\dagger}T^\prime)}= \sqrt{Tr(T^\dagger T)} = F_3(\zeta)
\end{equation}
Now, $ F_3(\zeta) = \sqrt{u_1 + u_2 + u_3} = \sqrt{ N(\zeta)^2-2 \sum \sqrt{u_i} \sqrt{u_j}}  \le N(\zeta) $,\\
$\sum \sqrt{u_i} \sqrt{u_j}$ being positive. \\
Therefore any quantum state which is 3-steerable i.e. $ F_3(\zeta) > 1 $, is useful for teleportation.

\subsubsection{From joint measurability to steering inequality}

Except from quantum entanglement, another necessary ingredient which is necessary for study of quantum nonlocality is the existence of incompatible set of measurements. In the simplest bipartite scenario Wolf \emph{et al.} have shown that any set of two incompatible POVMs with binary outcomes can always lead to violation of the CHSH-Bell inequality \cite{Wolf2009}. But, recently in refs.\cite{bru,ghu} the authors have proved that this result does not hold in the general scenario where numbers of POVMs and outcomes are arbitrary. However in this general settings the authors of \cite{bru,ghu} have established a connection between measurement incompatibility and a weaker form of quantum nonlocality i.e., EPR-Schr\"{o}dinger steering. They have shown that for any set of  incompatible POVMs (i.e. not jointly measurable), one can find an entangled state, such that the resulting statistics violate a steering inequality. Note that, it has been recently proved that the connection between measurement incomparability and steering holds for a more general class of tensor product theories rather than just Hilbert space quantum theory \cite{Manik'2015}. 

Let Alice perform a measurement assemblage $\{A_{a|x}\}$ on her part of a bipartite shared quantum state $\rho_{AB}$. Upon performing measurement $x$, and obtaining outcome $a$, the (un-normalized) state held by Bob is given by $\sigma_{a|x}=\mbox{Tr}(A_{a|x}\otimes\mathbf{1}\rho_{AB})$. The normalized state on Bob's side is given by $\sigma_{a|x}/\mbox{Tr}(\sigma_{a|x})$. Also we have $\sum_a\sigma_{a|x}=\sum_a\sigma_{a|x'}~\mbox{for}~x\ne x'$, which actually ensure no signaling from Alics to Bob. The state assemblage $\{\sigma_{a|x}\}$ is unsteerable iff it admits a decomposition of the form
\begin{equation}
\sigma_{a|x}=\pi(\lambda)p(a|x,\lambda)\sigma_{\lambda},~~\forall~a,x, \label{steer}
\end{equation} 
where $\sum_{\lambda}\pi(\lambda)=1$. Existence of such decomposition for state assemblage on Bob's side ensures that the statistics obtained from the state $\rho_{AB}$ admit a combined LHV-LHS model of the form of Eq.(\ref{steer}). The authors in refs.\cite{bru,ghu} have shown that the assemblage $\{\sigma_{a|x}\}$, with $\sigma_{a|x}=\mbox{Tr}(A_{a|x}\otimes\mathbf{1}\rho_{AB})$, is unsteerable for any state $\rho_{AB}$ acting on $\mathbb{C}^d\otimes\mathbb{C}^d$ if and only if the set of POVMs $\{A_{a|x}\}$ acting on $\mathbb{C}^d$ are jointly measurable. As a result we can say that

\begin{lemma}\label{assemblage}
	The assemblage $\{ \sigma_{a|x} \}$, with $\sigma_{a|x} = \tr_A(A_{a|x} \otimes \mathbf{1} \rho_{AB} )$ is unsteerable for any state $\rho_{AB}$ acting in $ \mathbb{C}^d \otimes \mathbb{C}^d$ if and only if the set of POVMs $\{ A_{a|x} \}$ acting on $\mathbb{C}^d$ are jointly measurable.
\end{lemma}
We are now in a position to present the result, which is described in the following theorem.
\begin{theorem}
	 Consider a composite quantum system composed of two subsystem with state spaces $\mathcal{H}_1$ and $\mathcal{H}_2$, respectively. In two qubits ,for any trio of dichotomic observables $A_1$, $A_2,A_3$ for the first system and three dichotomic observables $B_1$, $B_2,B_3$ for the second system and the joint state $\rho_{AB}$ acting on $\mathcal{H}_1\otimes\mathcal{H}_2$, we have the following inequality:
	\begin{equation}
	F_3 \le \frac{1}{\eta_{opt}},
	\end{equation}
	where $\eta_{opt}$ is the optimal unsharpness parameter that allows joint measurement for any three dichotomic quantum observables.($ F_3 $ refers to the violation of inequality (\ref{steering 3}))
\end{theorem}	

\begin{proof}

 Let us consider two arbitrary dichotomic observables $\{A_{a|x}\}$ on Alice's side, $x\in\{1,2,3\}$ and $a\in\{-1,+1\}$. These two observables in general may not allow joint measurement. However, introduction of unsharpness makes it possible to measure the unsharp versions of these two observables jointly. Let the optimal unsharpness be $\eta_{opt}$ which allows joint measurement for any two dichotomic observables.  
	
	Now according to lemma (\ref{assemblage}), as far as observables on Alice's side are jointly measurable, they will not violate any steering inequality and hence the steering inequality
	
	\begin{equation}
	F_3(\rho,\mu) =\frac{1}{\sqrt{3}} \Big|\sum_{i=1}^3\langle A_i^\eta\otimes B_i\rangle\Big| \leqslant 1,
	\end{equation}
	
	\begin{equation}
	F_3(\rho,\mu) =\frac{1}{\sqrt{3}} \Big|\sum_{i=1}^3 \eta \langle A_i \otimes B_i\rangle\Big| \leqslant 1,
	\end{equation}
	as, \begin{equation}
	\langle A^{(\eta)}_kB_j\rangle_{\rho_{AB}}=\eta\langle A_kB_j\rangle_{\rho_{AB}}.
	\end{equation} 
	The value of $\eta_{opt}$ in quantum theory is proved to be $1/\sqrt{3}$. Therefore the upper bound of the steering inequality (\ref{steering 3}) in quantum theory is $\sqrt{3}$.
	
\end{proof}

\section{Characterization in terms of spectrum}\label{sec:charac}
In this section we will derive the condition for absolute non-violation with respect to the inequality (\ref{steering 3}), where $ F_3 > 1 $ implies that the state is steerable in the 3 measurement scenario. We do not consider inequality (\ref{steering 2}) as under two measurement settings it is equivalent to the Bell-CHSH inequality \cite{caval16}, regarding which we have already derived results in \cite{bellGU1,bellGU2}.\\
Before proceeding with the final derivation, we prove the following two lemmas which are analogous to the proofs used in \cite{Wolf2002} with respect to the Bell operator. Pertaining to inequality (\ref{steering n}), we denote $ \mathfrak{S} = \frac{1}{\sqrt{n}} \Big|\sum_{i=1}^n\langle A_i\otimes B_i\rangle\Big| $ and by $ F_n $ the corresponding violation.\\
\begin{lemma}
$\mathfrak{S}$ is diagonal in the Bell-basis
\end{lemma} 
\begin{proof}
Let $ X = s_j \otimes I , (j=1,2,3)$.We find $ \langle A_i\otimes B_i\rangle_{X_j} =0$. Similarly letting $ Y = I \otimes s_j $, we observe $ \langle A_i\otimes B_i\rangle_{Y_j} =0$.
\end{proof}
We now use Theorem 4 from \cite{Wolf2002} in the steering scenario. 
\begin{lemma}
For any given spectrum of the density matrix, the respective Bell diagonal state  $ \zeta_B $ maximizes the steering violation. In other words $ F_{n}(\zeta_{B}) \ge F_{n}(U \zeta_{B} U^\dagger) $.   
\end{lemma} 
\begin{proof}
Retracing the steps of the proof in \cite{Wolf2002} we observe 
$Tr(U \zeta_{B} U^\dagger \mathfrak{S} ) \le Tr(\zeta_{B} \mathfrak{S} ) $.
\end{proof}
Therefore, the Bell diagonal states are optimal with respect to global unitary operations under steering scenarios \ref{steering n}.\\
The criteria now follows as,
\begin{theorem}
A state $ \sigma \in \mathbf{AUS}_3 $ iff $ 3Tr(\sigma^2)-2(x_1 x_2+x_1 x_3+x_1 x_4+x_2 x_3+x_2 x_4+x_3 x_4) \le 1 $. Here $ x_i $ are the eigenvalues of $ \sigma $. \label{AUS}
\end{theorem}
\begin{proof}
Any state can be brought to the respective Bell-diagonal state by a global unitary. Now, $ F_3(\zeta_B)^2 = 3(x_1^2 + x_2^2 + x_3^2 + x_4^2)-2(x_1 x_2+x_1 x_3+x_1 x_4+x_2 x_3+x_2 x_4+x_3 x_4) $. Hence the theorem.
\end{proof}
Similar to the case of absolutely separable states \cite{moor2001,lewinstein1998} an estimation of the size of the set $ \mathbf{AUS}_3 $can be obtained as below:
\begin{corollary}
(Estimation of the size of $ \mathbf{AUS}_3 $) A state $ \sigma \in \mathbf{AUS}_3 $ iff its purity is less than or equal to 1/2.
\end{corollary}
\begin{proof}
	The criterion in theorem (\ref{AUS}) can be recasted in the form $ 4 Tr(\sigma^2) \le 2 $. This also gives an estimation of the size of the ball around the maximally mixed state in which all the states $ \in \mathbf{AUS}_3 $. In the Frobenius norm the ball is expressed as $ \vert \vert \sigma - I/4 \vert \vert \le 1/2 $.
\end{proof}
The above condition guarantees that absolute non-violation ($ \mathbf{AUS}_3 $) can be also be verified by a necessary and sufficient condition in terms of bloch parameters. Consider the state $ \sigma $ Hilbert-Schmidt form (\ref{bloch}). Then the purity of any state$ \sigma $ can be expressed as,
\begin{equation}
Tr(\sigma^2) = \frac{1}{4}(1+ \vert \vert \overrightarrow{e} \vert \vert^2 + \vert \vert \overrightarrow{f} \vert \vert^2 + \vert \vert G \vert \vert^2)
\end{equation}
where $ \overrightarrow{e}, \overrightarrow{f} $ are the local bloch vectors of $ \sigma $ and $ G $ its correlation matrix.\\
Therefore, $ \frac{1}{4}(1+ \vert \vert \overrightarrow{e} \vert \vert^2 + \vert \vert \overrightarrow{f} \vert \vert^2 + \vert \vert G \vert \vert^2) \le 1/2  $, if and only if $ \sigma \in \mathbf{AUS}_3 $. Hence, 15 measurements are required to verfy that a state $ \in \mathbf{AUS}_3 $. This is in contrast to the fact that in order to check unsteerabilty w.r.t inequality (\ref{steering 3}) only 9 measurements are required.\\\
\indent We also note the following theorem pertaining to the unsteerability of the reduced subsystems of a three qubit pure state.
\begin{theorem}
The reduced state $ \sigma_{AB} $ of any pure three-qubit state $ \vert \Upsilon \rangle_{ABC} $ belongs to $ \mathbf{AUS}_3 $ if and only if $ \sigma_{C} $ is the maximally-mixed state.
\end{theorem} 
\begin{proof}
Consider that $ l= \vert \overrightarrow {\textbf{l}}  \vert$ ,where $ \overrightarrow {\textbf{l}} $ is the bloch vector for $ \sigma_{C} $. Then the eigenvalues of $ \sigma_{AB} $ are $ \lbrace (1+l)/2 , (1-l)/2 , 0 , 0 \rbrace $. Hence, in view of theorem (\ref{AUS}) , $ \sigma_{AB} \in \mathbf{AUS}_3  $  iff $ l $ vanishes. Therefore, as a consequence if all the local bloch vectors pertaining to subsystems $ A,B,C $ are zero , then $ \sigma_{AB}, \sigma_{BC},\sigma_{AC} \in \mathbf{AUS}_3 $. 
\end{proof}
It is also of interest to also note that if Alice,Bob and Charlie share a pure 3-qubit state, with nonzero local bloch vectors, then any two of the parties can "cheat" and collaborate, via application of of some unitary, to be able to violate the steering inequality. 
\section{Illustrations}\label{sec:illus}
 We now provide some illustrations on the application of our criterion. \\
(i) \textit{Absolutely Separable states -} Any absolutely separable state will $ \in \mathbf{AUS}_3 $. This is because , absolutely separable states preserve their separability under global unitary operation. In set theoretic language, if we denote the set containing the absolutely separable states by $\mathbf{AS}$, then $\mathbf{AS}$ forms a subset of $\mathbf{AUS}_3$. \\
(ii) Any pure product state cannot be in  $ \mathbf{AUS}_3 $ as they can always be converted to a pure entangled state by some global unitary. \\
(iii) \textit{Werner states-} The Werner state is given as \cite{werner89}, $ \sigma_{wer} = p \vert \psi^{-} \rangle \langle \psi^{-} \vert + \frac{1-p}{4} I$($ \vert \psi^{-} \rangle = \frac{\vert 01 \rangle - \vert 10 \rangle }{\sqrt{2}} $). The state is absolutely separable for $ p \le 1/3 $, hence $ \in \mathbf{AUS}_3 $  in that range. The eigenvalues are $ \lbrace (1+3p)/4 , (1-p)/4 , (1-p)/4 , (1-p)/4 \rbrace $. For $ p \le 1/\sqrt{3} $ it is in $ \in \mathbf{AUS}_3 $. The range is depicted in Fig (\ref{fig:werner}).
\begin{figure}
	\centering
	\includegraphics[width=0.7\linewidth]{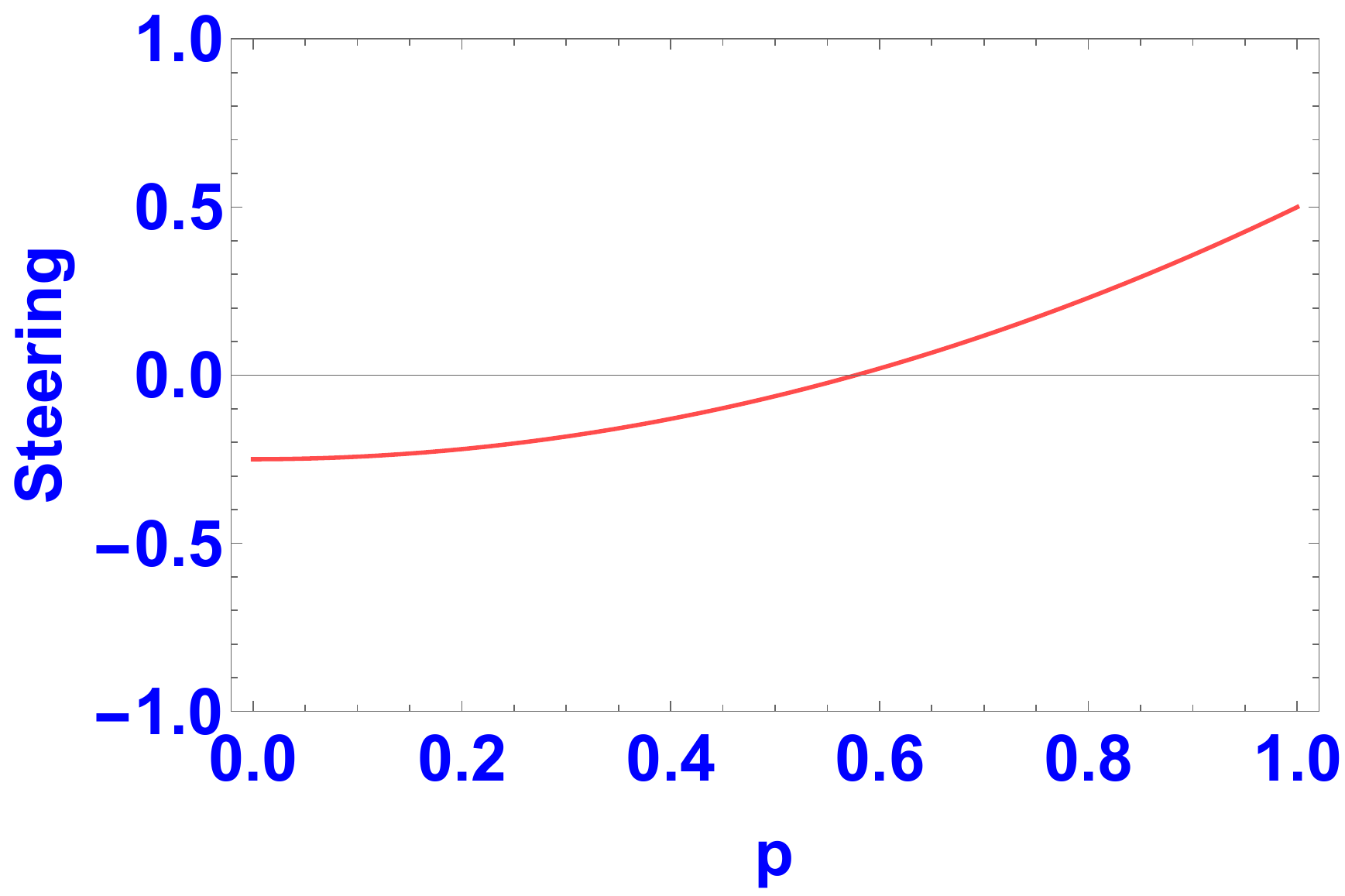}
	\caption{The vertical axis plots $ F_3-1 $.The curve cuts the $p$-axis for $p=\frac{1}{\sqrt{3}}$. So for any value of $p>=\frac{1}{\sqrt{3}},$ the Werner state is steerable after being subjected to suitable global unitary operations.}
	\label{fig:werner}
\end{figure}

(iv) \textit{Gisin states-} Gisin states were proposed in \cite{gisin96}.
Let $ \vert \psi_{\theta} \rangle  = sin \theta \vert 01 \rangle + cos \theta \vert 10 \rangle $ and $ \sigma_{mix} = \frac{1}{2}\vert 00 \rangle \langle 00 \vert + \frac{1}{2} \vert 11 \rangle \langle 11 \vert $ . Then the Gisin state is written as \cite{gisin96} ,
\begin{equation}
\sigma_{G} = \lambda \vert \psi_{\theta} \rangle  \langle \psi_{\theta} \vert + (1- \lambda) \sigma_{mix}
\end{equation}
with $ 0 < \theta < \pi/2, 0 \le \lambda \le 1 $. We observe that the Gisin states belong to $ \mathbf{AUS}_3 $ for $ \lambda \le 2/3 $.(See Fig(\ref{fig:gisin})) .\\
\begin{figure}
	\centering
	\includegraphics[width=0.7\linewidth]{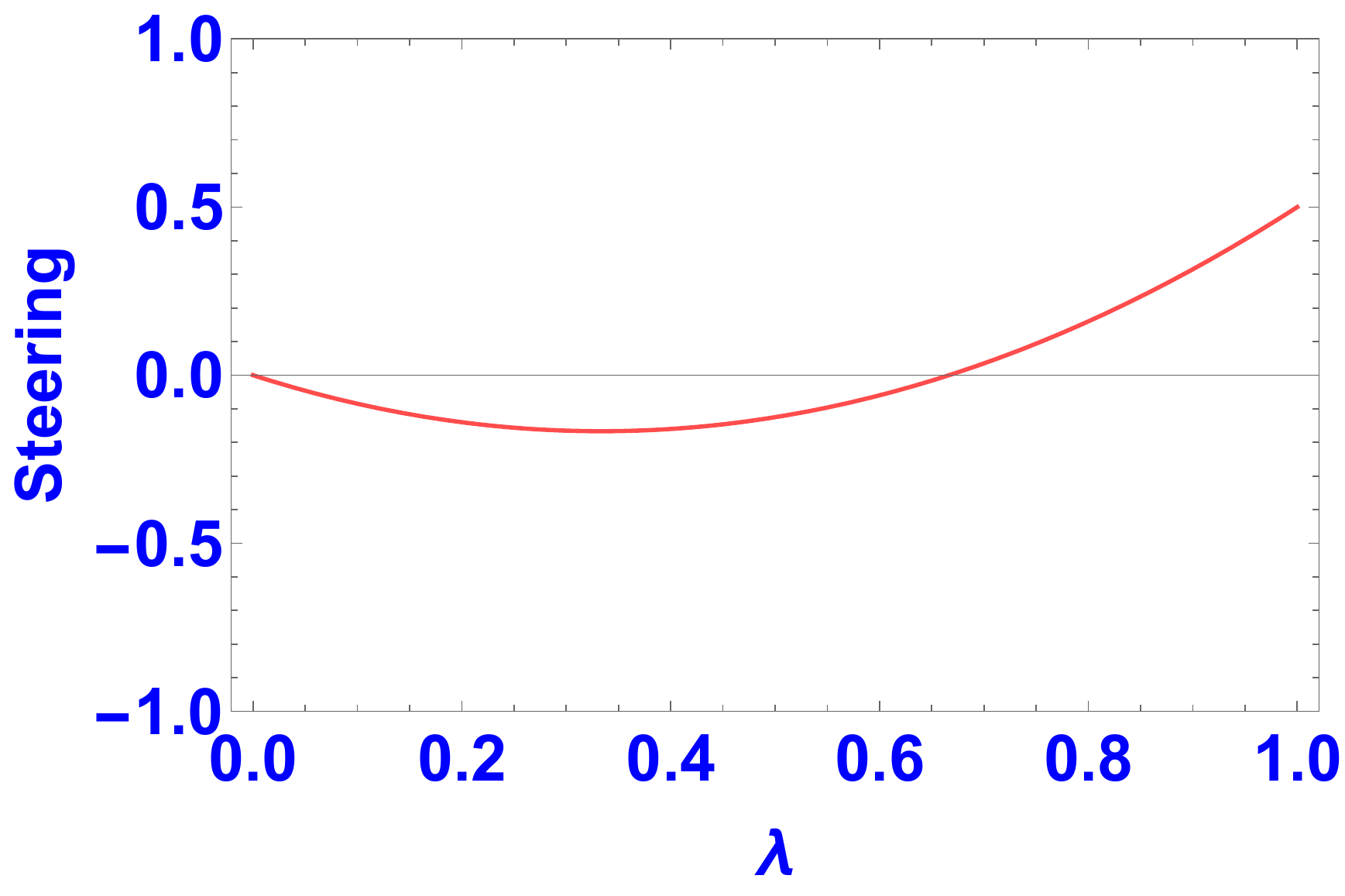}
	\caption{The vertical axis plots $ F_3-1 $.The portion of the curve above the positive direction of $\lambda$-axis indicates any Gisin state corresponding to this portion is steerable under action of suitable global unitary operations.}
	\label{fig:gisin}
\end{figure}

(v) \textit{$X$ states-} The $ X $ states are given as \cite{eberlyqic},
$ \sigma_{X} = v_1 \vert 00 \rangle \langle 00 \vert +v_2 \vert 01 \rangle \langle 01 \vert + v_3 \vert 10 \rangle \langle 10 \vert + v_4 \vert 11 \rangle \langle 11 \vert + v_5( \vert 00 \rangle \langle 11 \vert +  \vert 11 \rangle \langle 00 \vert) + v_6 ( \vert 01 \rangle \langle 10 \vert +\vert 10 \rangle \langle 01 \vert ) $. Here $ v_5^2 \le v_1 v_4 $ and $ v_6^2 \le v_2 v_3 $. According to our criterion,this state $ \in \mathbf{AUS}_3 $ iff $ \sum_{i=1}^{4} v_i^2 + 2(v_5^2 + v_6^2) \le 1/2$.\\
\section{Conclusion}\label{sec:concl}
\indent Steerability is a distinct notion of nonlocality weaker than Bell nonlocality but stronger than entanglement. Although envisaged by Schr\"{o}dinger, it was recently recasted in the form of a task in \cite{wiseman2013}. In the initial part of our present contribution we revisit steering inequalities and observe some typical features of steerability related to information processing tasks. \\
\indent The main focus of our work is to probe non-violation of some steering inequalities under global unitary action. Relations are derived between steerability and teleportation. We also analyse the optimal violation of steering inequalities under 3 measurement setting for two qubit systems.  Such states are characterized in terms of their spectrum and the size of the set to which they belong is also estimated. Interestingly we find the necessary and sufficient conditions for absolute non-violation w.r.t the inequality in terms of bloch parameters of the state.    Illustrations are provided to gain hands on insight to the impact of our work.\\
\indent Our work also opens future directions of research.In our work we have provided criterion for absoluteness of non-violation of a $3$-setting steering inequality by two-qubit states. There are steerable states which satisfy our criterion for absoluteness in terms of non-violation of the steering inequality. This implies that the proposed criterion does not imply that the given two-qubit state is also absolutely unsteerable in the $3$-setting
scenario. Therefore, one might be interested to probe absoluteness in terms of \textit{local hidden variable}--\textit{local hidden state} models for the measurement correlations arising from the given steering scenario.The extension of our work to other important steering inequalities in both bipartite and multipartite systems will be another area of useful investigation. 

\section*{Acknowledgement}

We would like to gratefully acknowledge fruitful discussions with Prof.Guruprasad Kar.AM acknowledges support from the CSIR project 09/093(0148)/2012-EMR-I.CJ acknowledges support through the Project SR/S2/LOP-08/2013 of the
DST, Govt. of India.


\begin{thebibliography}{99}
	
	
	\bibitem{Review Ent} R. Horodecki, P. Horodecki, M. Horodecki, and K. Horodecki, "Quantum entanglement", \href{http://journals.aps.org/rmp/abstract/10.1103/RevModPhys.81.865}{Rev. Mod. Phys. {\bf 81}, 865 (2009)}
	
	\bibitem{Review NL} N. Brunner, D. Cavalcanti, S. Pironio, V. Scarani, and S.Wehner, "Bell nonlocality",
	\href{http://journals.aps.org/rmp/abstract/10.1103/RevModPhys.86.419}{Rev. Mod. Phys. {\bf 86}, 839 (2014)}. 	
	
	\bibitem{Bell64} J. S. Bell, "On the Einstein Podolsky Rosen Paradox", Physics {\bf 1} (3): 195–200 (1964), J. S. Bell, Speakable and Unspeakable in Quantum Mechanics (Cambridge University Press, 1987).
	
	\bibitem{CHSH69} J.F. Clauser, M.A. Horne, A. Shimony, R.A. Holt, "Proposed experiment to test local hidden-variable theories", \href{http://dx.doi.org/10.1103/PhysRevLett.23.880}{Phys.Rev.Lett. {\bf 23}, 880(1969)}
	
	\bibitem{Bennett93} C. H. Bennett, G. Brassard, C. Cr\'{e}peau, R. Jozsa, A. Peres, and W. K. Wootters, "Teleporting an unknown quantum state via dual classical and Einstein-Podolsky-Rosen channels", \href{http://dx.doi.org/10.1103/PhysRevLett.70.1895}{Phys. Rev. Lett. {\bf 70}, 1895 (1993)}. 
	
	\bibitem{Bennett92}  C. H. Bennett, S. J. Wiesner, "Communication via one- and two-particle operators on Einstein-Podolsky-Rosen states", \href{http://dx.doi.org/10.1103/PhysRevLett.69.2881}{Phys. Rev. Lett. 69, 2881 (1992)}. 
	
	\bibitem{random} S. Pironio, A. Ac\'{i}n, S. Massar, A. Boyer de la Giroday, D. N. Matsukevich, P. Maunz, S. Olmschenk, D. Hayes, L. Luo, T. A. Manning and C. Monroe, " Random numbers certified by Bell’s theorem", \href{http://www.nature.com/nature/journal/v464/n7291/full/nature09008.html}{Nature {\bf 464}, 1021(2010)}. 
	R. Colbeck and R. Renner, ``Free randomness can be amplified", \href{http://www.nature.com/nphys/journal/v8/n6/full/nphys2300.html}{Nat. Phys.{\bf 8}, 450 (2012)};
	A. Chaturvedi and M. Banik, ``Measurement-device–independent randomness from local entangled states",
	\href{http://iopscience.iop.org/article/10.1209/
		0295-5075/112/30003/meta;jsessionid=D6C96ABB3E61
		C42C542A9553E8A4F4DC.c3.iopscience.cld.iop.org}{EPL {\bf 112}, 30003 (2015)}.
	
	\bibitem{key} J. Barrett, L. Hardy, and A. Kent, ``No signaling and quantum key distribution", \href{http://dx.doi.org/10.1103/PhysRevLett.95.010503}{Phys. Rev. Lett. {\bf 95}, 010503 (2005)};
	A. Ac\'{i}n, N. Gisin, and L. Masanes, ``From Bell's theorem to secure quantum key distribution", \href{http://dx.doi.org/10.1103/PhysRevLett.97.120405}{Phys. Rev. Lett. {\bf 97}, 120405 (2006)};
	
	\bibitem{dw} N. Brunner, S. Pironio, A. Ac\'{i}n, N. Gisin, A. A. Methot, and	V. Scarani, ``Testing the dimension of Hilbert spaces", \href{http://dx.doi.org/10.1103/PhysRevLett.100.210503}{Phys. Rev. Lett. {\bf 100}, 210503 (2008)};
	R. Gallego, N. Brunner, C. Hadley, and A. Ac\'{i}n, ``Device independent tests of classical and quantum dimensions",  \href{http://dx.doi.org/10.1103/PhysRevLett.105.230501}{Phys. Rev. Lett. {\bf 105}, 230501 (2010)};
	S. Das, M. Banik, A. Rai, MD R. Gazi, and S.Kunkri, ``Hardy's nonlocality argument as a witness for postquantum correlations",
	\href{https://journals.aps.org/pra/abstract/10.1103/PhysRevA.87.012112}{Phys. Rev. A {\bf 87}, 012112 (2013)};
	A. Mukherjee, A. Roy, S. S. Bhattacharya, S. Das, Md. R. Gazi, and M. Banik, ``Hardy's test as a device-independent dimension witness",
	\href{https://journals.aps.org/pra/abstract/10.1103/PhysRevA.92.022302}{Phys. Rev. A {\bf 92}, 022302 (2015)};
	
	\bibitem{game} N. Brunner and N. Linden, ``Connection between Bell nonlocality and Bayesian game theory", \href{http://www.nature.com/ncomms/2013/130703/ncomms3057/full/ncomms3057.html#references}{Nature Communications {\bf 4}, 2057 (2013)}.
	A. Pappa \emph{et al.} "Nonlocality and Conflicting Interest Games", 
	\href{http://journals.aps.org/prl/abstract/10.1103/PhysRevLett.114.020401}{Phys. Rev. Lett. {\bf 114}, 020401 (2015)}.
	A. Roy, A. Mukherjee, T. Guha, S. Ghosh, S. S. Bhattacharya, M. Banik, ``Nonlocal correlations: Fair and Unfair Strategies in Bayesian Game'', \href{https://doi.org/10.1103/PhysRevA.94.032120}{Phys. Rev. A {\bf 94}, 032120 (2016)}.
	
		\bibitem{epr} A. Einstein, B. Podolsky, and N. Rosen, 
		\href{http://journals.aps.org/pr/abstract/10.1103/PhysRev.47.777}{ Phys. Rev. {\bf 47}, 777 (1935)}.
		
		\bibitem{schr} E. Schr\"{o}dinger,  
		{ Proc. Cambridge Philos. Soc. {\bf 31}, 553 (1935); {\bf 32}, 446 (1936)}.
		
		\bibitem{wisemanprl07} H. M. Wiseman, S. J. Jones, and A. C. Doherty, 
		\href{https://journals.aps.org/prl/abstract/10.1103/PhysRevLett.98.140402}{ Phys. Rev. Lett. {\bf 98}, 140402 (2007)}.
		
		\bibitem{wisemanpra07} S. J. Jones, H. M. Wiseman, A. C. Doherty, 
		\href{http://journals.aps.org/pra/pdf/10.1103/PhysRevA.76.052116}{ Phys. Rev. A. {\bf 76}, 052116 (2007)}.
		
		\bibitem{jevtic2014} Antony Milne, Sania Jevtic, David Jennings, Howard Wiseman, Terry Rudolph, 
		\href{http://dx.doi.org/10.1088/1367-2630/16/8/083017}{ New. J. Phys. {\bf 16}, 083017 (2014)}
		
		\bibitem{rudolph2014} Sania Jevtic, Matthew Pusey, David Jennings,Terry Rudolph, 
		\href{http://journals.aps.org/prl/pdf/10.1103/PhysRevLett.113.020402}{ Phys. Rev. Lett. {\bf 113}, 020402 (2014)}.
		
		\bibitem{jevtic12014} Sania Jevtic, Michael J. W. Hall, Malcolm R. Anderson, Marcin Zwierz, Howard M. Wiseman, 
		\href{http://arxiv.org/abs/1411.1517v1}{ arXiv  {\bf 1411.1517v1} (2014)}.
		
		\bibitem{brunneroneway2014} Joseph Bowles, Tamás Vértesi, Marco Túlio Quintino, and Nicolas Brunner, 
		\href{http://dx.doi.org/10.1103/PhysRevLett.112.200402}{ Phys. Rev. Lett. {\bf 112}, 200402 (2014)}
		
		\bibitem{wittman2014} Bernhard Wittmann, Sven Ramelow, Fabian Steinlechner, Nathan K. Langford, Nicolas Brunner, Howard Wiseman, Rupert Ursin, Anton Zeilinger, 
		\href{http://dx.doi.org/10.1088/1367-2630/14/5/053030}{ New.J. Phys. {\bf 14}, 053030 (2012)}
		
		\bibitem{wiseman2014} D. A. Evans, H. M. Wiseman, 
		\href{http://dx.doi.org/10.1103/PhysRevA.90.012114}{ Phys. Rev. A. {\bf 90}, 012114 (2014)}
		
		\bibitem{reid} M.D.Reid, 
		\href{http://dx.doi.org/10.1103/PhysRevA.40.913}{ Phys. Rev. A. {\bf 40}, 913 (1989)}.
		
		\bibitem{ou} Z. Y. Ou, S. F. Pereira, H. J. Kimble, and K. C. Peng, 
		\href{http://dx.doi.org/10.1103/PhysRevLett.68.3663}{ Phys. Rev. Lett. {\bf 68}, 3663 (1992)}.
		
		\bibitem{caval} E. G. Cavalcanti, S. J. Jones, H. M. Wiseman, and M. D. Reid, 
		\href{10.1103/PhysRevA.80.032112}{ Phys. Rev. A. {\bf 80}, 032112 (2009)}.
		
		\bibitem{gurvits1} L.Gurvits, \textit{Proceedings of the thirty-fifth annual ACM symposium on Theory of computing}, Eds.  L. L. Larmore and M. X. Goemans, 10 (2003).

\bibitem{peres} A. Peres,"Separability Criterion for Density Matrices",\href{http://dx.doi.org/10.1103/PhysRevLett.77.1413} {Phys. Rev. Lett. 77, 1413 (1996)}.
\bibitem{horodecki} M. Horodecki, P. Horodecki, R. Horodecki,"Separability of mixed states: necessary and sufficient conditions" , \href{doi:10.1016/S0375-9601(96)00706-2}{Phys. Lett. A  223, 1  (1996)}.	
\bibitem{bound}  P. Horodecki,"Separability criterion and inseparable mixed states with positive partial transposition" , \href{doi:10.1016/S0375-9601(97)00416-7}{Phys. Lett. A 232, 333 (1997)};
M. Horodecki, P. Horodecki,R. Horodecki,"Mixed-State Entanglement and Distillation: Is there a “Bound” Entanglement in Nature?",\href{http://dx.doi.org/10.1103/PhysRevLett.80.5239} {Phys. Rev. Lett. 80, 5239 
(1998)}.
\bibitem{terhal} B. M. Terhal,"Bell inequalities and the separability criterion" , \href{doi:10.1016/S0375-9601(00)00401-1}{Phys. Lett. A 271,  319 (2000)}.

\bibitem{review} O. Guhne, G. Toth,"Entanglement detection",\href{doi:10.1016/j.physrep.2009.02.004}{ Phys. Rep. 474, 1  (2009)}.

\bibitem{holmes}R. B. Holmes, {\it Geometric Functional Analysis and its Applications}, (Springer-Verlag,Berlin, 1975).

\bibitem{barbieri} M. Barbieri, F. De Martini, G. Di Nepi, P. Mataloni, G. M.
D'Ariano, C. Macchiavello,"Detection of Entanglement with Polarized Photons: Experimental Realization of an Entanglement Witness",\href{http://dx.doi.org/10.1103/PhysRevLett.91.227901}{ Phys. Rev. Lett. 91, 227901  (2003)}.

\bibitem{wiec} W. Wieczorek, C. Schmid, N. Kiesel, R. Pohlner, O. Guhne,  
H. Weinfurter,"Experimental Observation of an Entire Family of Four-Photon Entangled States",\href{http://dx.doi.org/10.1103/PhysRevLett.101.010503} {Phys. Rev. Lett. 101, 010503 (2008)}.



\bibitem{telwit} N. Ganguly, S. Adhikari, A. S. Majumdar, J. Chatterjee,"Entanglement Witness Operator for Quantum Teleportation",\href{http://dx.doi.org/10.1103/PhysRevLett.107.270501}{Phys. Rev. Lett. 107, 270501 (2011)}.

\bibitem{telwit2} S. Adhikari, N. Ganguly, A. S. Majumdar,"Construction of optimal teleportation witness operators from entanglement witnesses",\href{http://dx.doi.org/10.1103/PhysRevA.86.032315}{Phys. Rev. A 86, 032315 (2012)}.

\bibitem{telwitcom} M.-J. Zhao, S.-M. Fei, X. Li-Jost,"Complete entanglement witness for quantum teleportation",\href{http://dx.doi.org/10.1103/PhysRevA.85.054301} {Phys. Rev. A 85,054301 (2012)}.

\bibitem{bellwit} P.Hyllus , O. Guhne , D. Bru$\beta$, M.Lewenstein , "Relations between entanglement witnesses and Bell inequalities", \href{http://dx.doi.org/10.1103/PhysRevA.72.012321}{Phys. Rev. A 72, 012321 (2005)}.

\bibitem{horobell} R.Horodecki , P.Horodecki, M.Horodecki , "Violating Bell inequality by mixed spin-1/2 states: necessary and sufficient condition", \href{doi:10.1016/0375-9601(95)00214-N}{Phys. Lett. A 200, 340 (1995)}.

\bibitem{vers} F. Verstraete, K. Audenaert, and B. D. Moor , "Maximally entangled mixed states of two qubits", \href{http://dx.doi.org/10.1103/PhysRevA.64.012316}{Phys. Rev. A 64, 012316 (2001)}.

\bibitem{johnston} N.Johnston , "Separability from spectrum for qubit-qudit states", \href{http://dx.doi.org/10.1103/PhysRevA.88.062330}{Phys. Rev. A 88, 062330 (2013)}.

\bibitem{absppt} R.Hildebrand , "Positive partial transpose from spectra", \href{http://dx.doi.org/10.1103/PhysRevA.76.052325}{Phys. Rev. A 76, 052325 (2007)}.

\bibitem{witabs} N.Ganguly, J.Chatterjee, A.S. Majumdar , "Witness of mixed separable states useful for entanglement creation", \href{http://dx.doi.org/10.1103/PhysRevA.89.052304}{Phys. Rev. A 89, 052304 (2014)}.

\bibitem{horodecki1996} Ryszard Horodecki, Michal Horodecki and Pawel Horodecki,"Teleportation, Bell's inequalities and inseparability",\href{https://doi.org/10.1016/0375-9601(96)00639-1}{ Phys.Lett.A {\bf 222},21-25(1996)}.

\bibitem{angelo2016} A. C. S. Costa, R. M. Angelo, "Quantification of Einstein-Podolski-Rosen steering for two-qubit states", { Phys.Rev.A {\bf 93}, 020103(R),(2016)}


	\bibitem{Wolf2009} Michael M. Wolf, David Perez-Garcia, Carlos Fernandez, "Measurements Incompatible in Quantum Theory Cannot Be Measured Jointly in Any Other No-Signaling Theory", 
	\href{http://dx.doi.org/10.1103/PhysRevLett.103.230402}{ Phys. Rev.Lett. {\bf 103}, 230402 (2003)}.
	
	
	\bibitem{bru} Marco T\'{u}lio Quintino, Tam\'{a}s V\'{e}rtesi, and Nicolas Brunner, "Joint Measurability, Einstein-Podolsky-Rosen Steering, and Bell Nonlocality",
	\href{http://dx.doi.org/10.1103/PhysRevLett.113.160402}{ Phys. Rev. Lett. {\bf 113}, 160402 (2014)}.
	
	
	\bibitem{ghu} Roope Uola, Tobias Moroder, and Otfried Gühne, "Joint Measurability of Generalized Measurements Implies Classicality", 
	\href{http://dx.doi.org/10.1103/PhysRevLett.113.160403}{ Phys. Rev. Lett. {\bf 113}, 160403 (2014)}.
	
	\bibitem{Manik'2015} Manik Banik,"Measurement incompatibility and Schrödinger-Einstein-Podolsky-Rosen steering in a class of probabilistic theories",
	\href{http://dx.doi.org/10.1063/1.4919546}{	J. Math. Phys. 56, 052101 (2015)}.
	
	\bibitem{caval16} Parth Girdhar and Eric G. Cavalcanti,"All two-qubit states that are steerable via Clauser-Horne-Shimony-Holt-type correlations are Bell nonlocal",
	\href{https://doi.org/10.1103/PhysRevA.94.032317}{ Phys.Rev.A {\bf 94}, 032317 (2016)}
	\bibitem{bellGU1} A.Roy , S.S. Bhattacharya, A.Mukherjee , N.Ganguly, "Violation of Bell-CHSH Inequality Under Global Unitary Operations", \href{http://arxiv.org/abs/1608.04099}{ arXiv:1608.04099 }
	\bibitem{bellGU2} Nirman Ganguly, Amit Mukherjee, Arup Roy, Some Sankar Bhattacharya, Biswajit Paul,K.Mukherjee "Bell-CHSH Violation Under Global Unitary Operations: Necessary and Sufficient conditions``, \href{http://arxiv.org/abs/1611.05586 }{ arXiv:1611.05586  }
	
	\bibitem{patro} Subhasree Patro, Indranil Chakrabarty, Nirman Ganguly, "Non-negativity of conditional von Neumann entropy and global unitary operations ``, \href{http://arxiv.org/abs/1703.01059 }{ arXiv:1703.01059  }
	
	\bibitem{Wolf2002} Frank Verstraete and Michael M. Wolf,"Entanglement versus Bell Violations and Their Behavior under Local Filtering Operations",
	\href{https://doi.org/10.1103/PhysRevLett.89.170401}{ Phys.Rev.Lett {\bf 89}, 170401 (2002)}
	
	\bibitem{moor2001}  Frank Verstraete, Koenraad Audenaert, and Bart De Moor,"Maximally entangled mixed states of two qubits",\href{https://doi.org/10.1103/PhysRevA.64.012316}{ Phys. Rev. A {\bf64}, 012316}
	
	\bibitem{lewinstein1998} Karol Życzkowski, Paweł Horodecki, Anna Sanpera, and Maciej Lewenstein,"Volume of the set of separable states",\href{https://doi.org/10.1103/PhysRevA.58.883}{ Phys. Rev. A {\bf 58}, 883}



   \bibitem{werner89} R. F. Werner; \emph{Quantum states with Einstein-Podolsky-Rosen correlations admitting a hidden-variable model},  \href{https://doi.org/10.1103/PhysRevA.40.4277}{Phys. Rev. A {\bf 40}, 4277 (1989)}

	\bibitem{gisin96} N.Gisin;
	\href{http://dx.doi.org/10.1016/S0375-9601(96)80001-6}{ Phys. Lett. A. {\bf  210}, Issue 3, 151–156, 2619(1996)}
	
	\bibitem{eberlyqic} T. Yu and J. H. Eberly, Quantum Inform. and Comput.7, 459 (2007);
	
	\bibitem{wiseman2013} Eric G. Cavalcanti, Michael J. W. Hall, Howard M. Wiseman,"Entanglement verification and steering when Alice and Bob cannot be trusted", 
	\href{http://dx.doi.org/10.1103/PhysRevA.87.032306}{ Phys. Rev. A. {\bf 87}, 032306 (2013)}.
	
\end{thebibliography}
 \end{document}